\documentclass[10pt,conference]{IEEEtran}

\usepackage{subcaption}

\DeclareCaptionFont{eightpt}{\fontsize{8}{10}\selectfont}

\usepackage{pgfplots}
\pgfplotsset{compat=1.18}
\usepackage{microtype}
\usepackage{braket}
\newcommand{\ketbra}[2]{\ket{#1}\!\bra{#2}}
\newcommand{\codeParam}[1]{(\!(#1)\!)}
\newcommand{\stabParam}[1]{[\![#1]\!]}
\usepackage{makecell}
\usepackage{cellspace}
\usepackage{tikz}
\pgfplotsset{compat=1.18} \usepackage{amsmath,amssymb,amsthm}
\usepackage{etoolbox}

\usepackage{graphicx}
\usepackage{url}

\theoremstyle{definition}
\newtheorem{definition}{Definition}[section]
\newtheorem{theorem}[definition]{Theorem}

\newtheorem{algorithm}[definition]{Algorithm}
\newtheorem{example}[definition]{Example}
\usepackage{booktabs}
\usepackage{siunitx} \usepackage{multirow}
\usepackage{lineno}
\usepackage{pdfpages}

\usepackage{cite}
\usepackage{float}

\usepackage[thm]{algorithm}
\usepackage{algpseudocode}
\makeatletter
\let\c@algorithm\c@definition 
\makeatother

\usepackage{placeins}
\makeatletter
\newenvironment{protocol}[1][htb]{
    \renewcommand{\ALG@name}{Protocol}
    \begin{algorithm}[#1]
    }{\end{algorithm}}
\makeatother
\usepackage{rotating}
\usepackage{ifpdf}
\ifpdf
  \usepackage{hyperref}
\fi

\DeclareMathOperator{\Tr}{Tr}

\newcommand{\EE}{\mathcal{E}}

\newcommand{\HH}{\mathcal{H}}

\makeatletter

\newcounter{ALGsubline}[ALG@line]

\newcommand{\SubState}{\stepcounter{ALGsubline}{\footnotesize
\Statex\theALG@line.\arabic{ALGsubline}:}\hspace{\algorithmicindent}\ignorespaces
}

\newenvironment{subnumbers}{        \setcounter{ALGsubline}{0}
    \let\State\SubState }{        }

\makeatother

\title{Quantum Anonymous Secret Sharing with \\Permutation Invariant Codes}

\author{\IEEEauthorblockN{Varin Sikand}
\IEEEauthorblockA{\textit{Department of Computer Science} \\
\textit{The University of Texas at Dallas}\\
Richardson, TX, USA \\
varin.sikand@utdallas.edu}
\and
\IEEEauthorblockN{Andrew Nemec}
\IEEEauthorblockA{\textit{Department of Computer Science} \\
\textit{The University of Texas at Dallas}\\
Richardson, TX, USA \\
andrew.nemec@utdallas.edu}
}

\begin{document}

\maketitle

\begin{abstract}
Quantum secret sharing schemes are a family of quantum cryptographic protocols which provide secure quantum encodings, mapping one secret to multiple shares of information such that the original secret cannot be accessed without an authorized set of shares present for decoding.
In this work, we describe a protocol that enables sender-anonymity during the secret decoding process.
By using permutation-invariant QEC codes along with a set of anonymous quantum transmission algorithms, we construct a quantum anonymous secret sharing scheme that achieves sender-anonymity. 
We quantify information leakage in ramp quantum secret sharing schemes via the quantum conditional min-entropy, justifying it as a valid measure of leaked information by relating it to the Knill-Laflamme quantum error correction conditions.
Finally, we evaluate several permutation-invariant codes using this measure to make observations on the information leakage of intermediate shares for each quantum anonymous secret sharing scheme.
\end{abstract}

\section{Introduction}

Classical secret sharing (CSS) \cite{shamir1979share,blakley1979safeguarding} is a cryptographic primitive which allows for the division of a classical secret into shares distributed among some number of parties in such a way that certain subsets of parties can work together to recover the secret, while other subsets have no information about the secret.
This helps eliminate single points of failure, both in terms of potentially untrustworthy participants as well as the potential loss of a share. Applications range from facilitating secure multiparty computing protocols~\cite{CramerSecureMPC} to ensuring some minimal consensus when making critical organization-level decisions~\cite{barker2016nist}.

Quantum computing promises exciting new ways of processing classical and quantum information.
Shor's algorithm~\cite{Shor_1997} and its implications for public key cryptosystems have dominated the cryptography discussion, leading to the development of post-quantum cryptography~\cite{915326}.
While post-quantum protocols are primarily classical schemes that protect against quantum adversaries, a number of classical cryptographic protocols such as secret sharing have been adapted into the quantum paradigm.
Together with some more novel quantum schemes such as quantum key distribution~\cite{bennett1984,ekert}, these protocols make up quantum cryptography: protocols that make use of quantum properties such as superposition and entanglement.

Quantum secret sharing (QSS) was initially introduced by Hillery et al.~\cite{Hillery_1999} by translating CSS schemes to a quantum setting.
This has led to the development of QSS schemes for classical secrets~\cite{gottesman2000theory, Karlsson_QSS, Hillery_1999, Miyajima_2022} and quantum secrets~\cite{gottesman2000theory, hagiwara2020fourqubitscodequantumdeletion, Shibayama_2021}.
Recent advancements in hardware have also lead to the experimental realizations of these schemes~\cite{Chiwaki_2025, masumori2024advancesharingogawaet}.

However, while these schemes keep the secret secure from any unauthorized participants, they do not also allow for the identity of the participants to remain a secret.
Similar to applications such as anonymous voting, where participants require anonymity in order to participate honestly and without bias, some applications of secret sharing may demand the anonymity of those participants who recover the secret.
This need for anonymity has motivated the ongoing development of classical anonymous secret sharing schemes \cite{BishopFASS, AnonShamir2025, BLUNDO199713}, and similarly motivates the need for an anonymous secret sharing scheme within the quantum paradigm. 

In this work, we will utilize the class of permutation-invariant (PI) quantum error-correcting (QEC) codes, which have codespaces invariant under reorderings of the physical qubits, allowing for decoding processes that depend only on the number of erasures that have occurred and not specific subsets of erased qubits.
Using these PI codes, we present a QSS protocol that allows all participants of a recovering set to remain completely anonymous throughout and after the secret recovery step.
Notably, this protocol works for any PI code, allowing for easy construction of anonymous QSS schemes with varying numbers of participating shareholders. Similar recent works in quantum anonymous secret sharing (QASS) schemes lack this ability to tolerate missing shareholders, limiting their uses to those requiring full consensus as opposed to a threshold \cite{LI2024109836,WangNovel2024,Yang2025}. Additionally, these works do not keep shareholder identities secret from the decoder and therefore do not maintain sender-anonymity between the decoder and the participating shareholders.
We use sub-protocols from Christandl and Wehner~\cite{Christandl_2005} for anonymous quantum transmissions between a participant and the recover, obscuring which parties participate in the recovery from both external adversaries and other participating parties including the decoder. 

Furthermore, we present the quantum conditional min-entropy~\cite{Konig_2009} of a bipartite reference-experimental state as a sound metric for quantifying leakage of information in the event that an adversary gains access to a certain number of shares.
We justify the validity of this metric by connecting it to the Knill-Laflamme quantum error correction conditions~\cite{Knill_2000}, establishing it as a single-shot quantity measuring the amount of information lost between an attempted recovery and an optimal recovery.
This metric can be further applied to the class of ramp QSS schemes, which leak information about the secret for some sets of parties that are neither authorized not unauthorized.
We compare the conditional min-entropy of several PI codes, allowing us to observe information leakage in our QASS schemes.

\section{Background}

\subsection{Secret Sharing} \label{sec:sec-share}

Within the field of cryptography, there are many methods of keeping a secret secure, aiming to protect any information of the secret from adversaries. Encryption schemes, for example, attempt to protect from adversaries by storing the secret in a state that is unintelligible to those not authorized. However, what about when an authorized user later becomes malicious, either via adversary control over the user or due to the user's own intentions changing? In such cases, where no single user can be fully trusted, secret sharing schemes present a solution by not allowing any single user to decode the secret on their own. CSS schemes~\cite{shamir1979share,blakley1979safeguarding} perform this encoding using classical bits as secrets and shares with a wide range of established applications requiring consensus approval~\cite{CramerSecureMPC,barker2016nist}.

A QSS scheme is a cryptographic protocol which aims to encode a single secret into several shares stored as a collection of entangled qubits. These qubits can then be distributed among multiple parties such that the participation of some number of the shares is required to recover the original secret. These secrets can be either classical~\cite{gottesman2000theory, Karlsson_QSS, Hillery_1999, Miyajima_2022}, encoded as a mixture of basis states, or quantum~\cite{gottesman2000theory, hagiwara2020fourqubitscodequantumdeletion, Shibayama_2021}, preserving superposition. 

Some QSS schemes utilize multipartite entangled states as shares \cite{LI2024109836, Hillery_1999, e24101433, Karlsson_QSS}, however many are limited to encoding classical secrets \cite{Hillery_1999,Karlsson_QSS,e24101433} whereas others such as Guo-Dong et al. \cite{LI2024109836} require all shareholders to yield their shares in order to decode. In general, these forms of quantum secret sharing in general are very fragile to noise and decoherence, with much of the modern work in the field regarding these QSS schemes aiming to address this fragility \cite{LI2024109836,Yang2025,WangNoiseQSS}. 

It is also possible to make QSS schemes anonymous \cite{LI2024109836,WangNovel2024,Yang2025}. Current QASS schemes are generally restricted to GHZ or W states which tolerate either no or single qubit losses before information is lost, limiting the access structure. Additionally, these schemes distribute $m$ shares among $n$ participants but require all shareholders to yield their shares in order to decode the original secret, making it sub-optimal for partial consensus use cases. Yang et al. \cite{Yang2025} proposes a more robust QASS scheme that allows $n-m-2$ participants to be missing, however it still assumes that all $m$ participants who receive shares participate in recovery. Finally, all of these QASS schemes do not keep the senders' identities anonymous from the decoder due to the total shareholder participation constraint.

In this paper, we focus on QSS schemes that leverage QEC codes, specifically modeling the schemes as quantum codes under erasure channels, to provide more robust secret sharing protocols \cite{gottesman2000theory}. Unlike entanglement-based schemes, QEC-based approaches allow for secret recovery even with some number of missing shares while allowing a vast number of QEC codes to be interpreted as QSS schemes. 

In a QEC-based quantum secret sharing scheme, the secret is encoded using a QEC code where the post-error encoded states correspond to authorized sets within the scheme's access structure. This scheme allows for some share-holders to remain absent from the secret recovery process while still recovering full information of the secret. A common access structure defines authorized sets using a threshold $k$ such that any set of shares $S$ is authorized if and only if $|S| \geq k$ and sets of size $|S| < k$ obtain no information about the secret. This means, for these \textit{threshold schemes}, the threshold of such a secret sharing scheme naturally aligns with the erasure-correction capability of the code. That is, if $k$ shares are required to reconstruct the secret then the code must be able to correct $n-k$ erasures, where $n$ is the total number of shares \cite{gottesman2000theory}.

QSS schemes can be specified further as either \textit{perfect schemes} or \textit{ramp schemes} based on how they handle information leakage. In a perfect scheme, all share sets are distinctly either authorized or unauthorized. Clearly, these sorts of schemes are ideal in terms of the security it provides, however the share size for a perfect QSS scheme is constrained to be at least as large as the secret itself \cite{gottesman2000theory}. This restriction leads to a significant increase in the size of each share in general, and thereby less efficient storage and transportation of each individual share. To circumvent this, Ogawa et al. \cite{Ogawa_2005} proposed ramp secret sharing schemes which allow smaller individual share sizes at the cost of \textit{intermediate} share sets which leak information about the secret. 

We will focus on ramp QSS schemes for this paper in order to define a usable metric for information leaked by intermediate sets. Therefore, the threshold of a ramp secret sharing scheme can be viewed as a ``fully authorized'' set of qubits, or the set of qubits such that perfect information can be retrieved. Using this paradigm, we can model an adversarial party as one who gains access to some number of shares, and the conditional min-entropy defined in Section \ref{subsec:cond-min-entropy} can be seen as a measure of how much information is ``leaked'' to this adversary given their intermediate share set.

Certain CSS schemes leak less information than QSS schemes for the same number of shares held by an adversary. \textit{Hybrid} secret sharing schemes \cite{PhysRevA.64.042311, PhysRevA.71.012328} allow ramp QSS schemes to be paired with well-chosen CSS schemes, ``lifting'' the overall security in exchange for a linear increase in the number of shares per shareholder. By applying one of the four single-qubit Paulis uniformly at random to the quantum secret before encoding it, the dealer can then use two classical bits as secrets to identify which Pauli was applied. This effectively encodes the quantum secret using two classical secrets, and the adversary must then be able to first decode the classical secret in order to gain information about the quantum secret~\cite{892142,PhysRevA.67.042317}. This guarantees the security of the QSS scheme to be at least as good as the security of the CSS scheme for any number of shares held by the adversary.

\subsection{Quantum Error Correction}

Intuitively, a QEC code is defined as a subspace $\mathcal C$ of a Hilbert space $\mathcal{H}_p$ that can correct a set of error (Kraus) operators $\mathcal{E} = \{E_i\}$ via some recovery superoperator $\mathcal{R}$:
\[(\mathcal{R}\mathcal{E})\circ(\rho) \propto \rho.\]

Knill and Laflamme~\cite{Knill_2000} gave several equivalent conditions that define a QEC code, three of which we provide below:

\begin{definition}[Knill-Laflamme conditions]\label{def:Knill-Laflamme}\hfill
\begin{enumerate}
    \item A code $\mathcal{C}$ corrects a set of errors $\mathcal{E}$ if and only if, for all basis elements $\ket{i}, \ket{j} \in \mathcal C$ and error operators $E_a,E_b \in \mathcal{E}$:
    \begin{equation}\label{eq:knill-laflamme}
        \bra{i}E_a^\dagger E_b\ket{j} = \delta_{ij}\alpha_{ab},
    \end{equation}
    where $\alpha_{ab} \in \mathbb{C}$ is a constant dependent on only the two error operators.

    \item Let $R,E$ be identical and maximally entangled systems encoded into codespace $\mathcal{C}$ such that their joint system is
    \[\sum_{i} \ket{i}_{R}\ket{i}_{E}.\]
    Then, $\mathcal{R}$ is a perfect recovery channel for the error superoperator $\mathcal{E}$ on $\mathcal{C}$ if and only if 
    \begin{equation}\label{eq:Knill-Laflamme-2}
        I \otimes (\mathcal{RE}) \sum_{i}\ket{i}_{R}\ket{i}_{E} = \lambda\sum_{i} \ket{i}_{R}\ket{i}_{E},
    \end{equation}
    where $\lambda \in \mathbb{C}$.
    \item Given initially equivalent systems $R, E$ encoded in codespace $\mathcal{C}$, let the experimental state undergo the error channel $\mathcal{E}$ while the reference system is left untouched, resulting in systems $R, E'$. The codespace $\mathcal{C}$ represents an error-correcting code for $\mathcal{E}$ if and only if
     \begin{equation}\label{eq:knill-laflamme-3}
         I(\rho_{E'}:\rho_R) - I(\rho_E:\rho_R) = 0.
     \end{equation}
\end{enumerate}
\end{definition}
Throughout the paper, we will refer to these conditions as the first, second, and third Knill-Laflamme conditions respectively.

Quantum codes encode a logical state in $\HH_K$ into $n$ physical qubits and are defined by a minimum distance $d$, such that all errors affecting up to $d-1$ physical qubits can be detected by the code. We call such a code an $\codeParam{n,K,d}$ quantum code. This code can additionally correct all errors affecting up to $t = \lfloor \frac{d-1}{2}\rfloor$ qubits. The Shor code~\cite{PhysRevA.52.R2493}, for example, is a $\codeParam{9,2,3}$ code, meaning that it encodes one logical qubit into nine physical ones, and can detect up to two general errors. We denote them as ``general'' errors because, by default, we do not know the location of these errors.

While there are many possible error channels, complete depolarization (or ``loss'') of a qubit is a common type considered in the field of QEC. A complete depolarization can be modeled in one of two ways:
\begin{enumerate}
    \item Tracing out the lost qubit from the rest of the system
    \item Performing a completely depolarizing operation on the qubit, mapping it to the maximally mixed state
\end{enumerate}

The second operation can defined in terms of Kraus operators:
\begin{equation}\label{eq:depolar}
    \Phi(\rho) =  \frac{1}{4}\sum_{K \in P}K\rho K^\dagger = \frac{I}{2},
\end{equation}
where the set of Kraus operators is chosen to be the single-qubit Pauli set $P$ and $\rho$ represents the density matrix of the qubit. The depolarization of multiple qubits can be seen as a series of these channels, one per qubit.

By leveraging the discretization of errors, complete loss of information on a set of qubits can be viewed as $t$ Pauli errors to be corrected after measurement. Therefore, depolarization of $t$ qubits can be corrected for by an $\codeParam{n,K,d}$ code. However, if we are aware of the location of the depolarization, we can correct a larger set of errors. We define such an error paired with knowledge of its position as a \textit{quantum erasure}. In this case, any QEC code with distance $d$ can be considered a quantum erasure code that corrects $d-1$ erasures.

\subsection{Quantum Anonymous Transmission Protocols}
\label{sec:anon-tran}
In order to define a QASS protocol, we first define anonymity along with several prerequisite anonymous protocols for sending and receiving information. We assume the adversary attempting to obtain information regarding the shareholders' identity is unbounded and can \textit{corrupt} a subset of players, controlling all of their actions (a player being either a shareholder or the decoder). A player may separately choose to not adhere to the protocol, denoted in this case as \textit{malicious}~\cite{Christandl_2005}. Finally, we define an \textit{honest} player as a player that is neither corrupted nor malicious.

We now define a protocol as \textit{anonymous} when the protocol does not change the a priori uncertainty of any sender’s identity given an adversary who can corrupt up to $n-2$ players and gains access to the randomness of all corrupted players (randomness meaning the set of all random strings generated during the protocol for corrupted players). \textit{Tracelessness} is a stronger constraint, additionally asserting that if, after the protocol is completed, the adversary gains access to all randomness used by players during the protocol, the probability of determining the sender or receiver’s identity is unchanged.

Transmitting quantum information is commonly done via quantum teleportation or a quantum channel, however both would break anonymity (either via entanglement or a channel being set between two parties)~\cite{BennetTeleportation}. 
Therefore, to maintain anonymity during this process, we incorporate into our QSS scheme the following anonymous and traceless protocols for transmitting both classical and quantum information introduced by Christandl and Wehner~\cite{Christandl_2005}  :
\begin{enumerate}
    \item Creating entanglement between two parties among $n$ total parties anonymously (\textit{AE}, Protocol \ref{prot:AE}). 

    \textit{Input:} One $n$-qubit GHZ state $\frac{1}{\sqrt2}(\ket{0^n} + \ket{1^n})$.
    
    \textit{Output:} Entangled bell pair $\frac{1}{\sqrt{2}}(\ket{00} + \ket{11})$ between sender and receiver.
    \item Broadcasting a classical bit $d$ anonymously (\textit{ANON(d)}), Protocol \ref{prot:ANON}).

    \textit{Input:} One $n$-qubit GHZ state $\frac{1}{\sqrt2}(\ket{0^n} + \ket{1^n})$.
    
    \textit{Output:} One classical bit $d$ sent from sender to receiver.
    
    \item Sending quantum information anonymously via quantum teleportation (\textit{ANONQ($\ket{\phi}$)}, Protocol \ref{prot:ANONQ}).

    \textit{Input:} Three $n$-qubit GHZ states $\frac{1}{\sqrt2}(\ket{0^n} + \ket{1^n})$.
    
    \textit{Output:} One qubit $\ket{\phi}$ sent from sender to receiver
    
    \item For simultaneous send requests (either entanglement or \textit{ANON}), we use the collision detection protocol defined in \cite{Christandl_2005} and a reservation map strategy described in \cite{Waidner1990}. Uses $\lceil \log n\rceil + 1$ $n$-qubit GHZ states.
\end{enumerate}

\begin{protocol}[H]
\caption{\textit{AE}}
\label{prot:AE}
    Prerequisite: Shared state $\frac{1}{\sqrt{2}}(\ket{0^n} + \ket{1^n})$
\begin{algorithmic}[1]
    \State Alice and Bob do nothing to their qubits, while the other $n-2$ players (let $V$ be the set of all players) apply a Hadamard to their qubit
    \State These other players then measure their qubit and broadcast the result $m_j$
    \State Alice picks a random bit $b \in_R \{0,1\}$ and broadcasts $b$
    \State Alice applies a phase flip $Z$ to her qubit if $b=1$
    \State Bob applies a phase flip $Z$ to his qubit if $b\,\oplus \bigoplus_{j\in V\backslash\{Alice,Bob\}} m_j = 1$
\end{algorithmic}
\end{protocol}\vspace{-1em}

\begin{protocol}[H]
\caption{\textit{ANON(d)}}
\label{prot:ANON}
    Prerequisite: Shared state $\frac{1}{\sqrt{2}}(\ket{0^n} + \ket{1^n})$
\begin{algorithmic}[1]
    \State Alice applies a phase flip $Z$ to her part of the state if $d$ = 1, otherwise nothing
    \State Each player (including Alice):
    \begin{algorithmic}[1]
        \State Applies a Hadamard transform to their qubit
        \State Measures their qubit in the computational basis
        \State Broadcasts the result of that measurement
        \State Counts the total number of 1’s in the measurement outcomes. Let this be $k$
        \State If $k$ is even, then $d = 0$, otherwise $d = 1$
    \end{algorithmic}
\end{algorithmic}
\end{protocol}\vspace{-1em}

\begin{protocol}[H]
\caption{\textit{ANONQ($\ket{\phi}$)}}
\label{prot:ANONQ}
Prerequisite: 3 shared states $\frac{1}{\sqrt{2}}(\ket{0^n} + \ket{1^n})$. Alice has qubit $\ket{\phi}$ as her share to send.
    \begin{algorithmic}[1]
        \State The players run \textit{AE} as described above such that Alice and Bob then share an EPR pair
        \State Alice uses the quantum teleportation circuit with input $\ket{\phi}$ (the qubit she wishes to send) and the EPR pair, and obtains measurement outcomes $m_0$ and $m_1$.
        \State All players run \textit{ANON($m_0$)} and \textit{ANON($m_1$)} with Alice as the sender
        \State Bob applies the transformation described by $m_0,m_1$ on his part of the EPR pair and obtains $\ket\phi$
    \end{algorithmic}
\end{protocol}\vspace{-1em}

\section{Quantifying Information Leakage}

\subsection{Quantum Mutual Information}
For QEC codes, the information leakage, i.e., the information inaccessible to the receiver, is typically framed in terms of quantum mutual information between an experimental state undergoing the error and an untouched reference state.
This is based on the third Knill-Laflamme condition (Eq. \ref{eq:knill-laflamme-3}), which states that the mutual information between such experimental/reference states must be equal both before and after the error for perfect recovery. Equivalently, Schumacher and Nielsen \cite{Schumacher_1996} define this condition for perfect error correction in terms of \textit{coherent information} $I_e$, stating that the coherent information after the error channel $\EE$ must equal the entropy pre-error,
\begin{equation*}
    I_e(\rho_E, \mathcal{E}) := S(\rho_{E'}) - S(\rho_{E'R}) = S(\rho_E),
\end{equation*}
where the coherent information is dependent on only $\rho_E$ and the error channel $\mathcal{E}$, as $E'$ can be define in terms of both.

Quantum mutual information is therefore effective in determining if a state is fully recoverable (and as shown in Ogawa et al. \cite{Ogawa_2005}, fully vanishing), but in general it is less useful in determining how much information can be recovered in an imperfect recovery scenario.
Additionally, mutual information is not a ``single-shot'' measure, as the entropies used in its calculation are an average over the density matrix's spectrum.
In other words, mutual information represents information gathered from multiple copies of the state and averages over them.
While this makes it useful in the context of quantum error correction, it does not accurately represent an adversary attempting to retrieve information about a secret encoded in a single instance of a QSS scheme, as they would only have access to one copy of the shares.

\subsection{Conditional Min-Entropy} \label{subsec:cond-min-entropy}
The first Knill-Laflamme condition is useful when determining if a codespace is an error-correcting code that can correct a specific set of errors.
In our case, however, we wish to establish a quantity that describes how ``recovered'' a post-error state is after applying an optimal recovery channel, which the second Knill-Laflamme condition (Eq. \ref{eq:Knill-Laflamme-2}) provides.
The condition states that if the experimental system $E$ remains fully entangled with the reference system $R$ after undergoing an error and a recovery process, then $\mathcal{R}$ represents a valid recovery superoperator for the error set $\mathcal{E}$ and codespace $\mathcal{C}$. This is particularly useful when considering imperfect recovery processes, where not all the information can be recovered.

If a fully recovered state is one that remains fully entangled with its reference system, then we can measure the fidelity between our combined reference-experimental state $\rho_{RE}$ post error and recovery with the maximally entangled state $\ket{\Phi_{RE}} =\frac{1}{\sqrt{d_R}}\sum_{i}\ket{i}_{R}\ket{i}_{E}$, where $d_R=\dim(R)$.
As stated in the condition, if this fidelity is $1$, then the recovery operation has fully recovered all information, however when a system has lost all information, it can still utilize the maximally mixed state $I/d^2_R$ and obtain an fidelity of $1/d_R^2$.
This value, therefore, indicates that the recovery operation was entirely unable to recover the combined systems to an entangled state since the experimental system is now no better than a maximally mixed state. 

To measure this fidelity, Konig et al. \cite{Konig_2009} defined the conditional min-entropy $H_{min}(R|E)$ as a measure of the maximum achievable overlap between two systems $R$ and $E$, under the assumption that only local operations are performed on $E$.
This mirrors the setup surrounding the second Knill-Laflamme condition, with $R$ representing the reference system and $E$ representing the experimental system post-error.
This ``single-shot'' measurement also models the optimal recovery operation even when the operation itself is unknown, as $H_{min}$ is defined over all possible local operations on subsystem $E$.
The fidelity between two density operators $\rho$ and $\sigma$ is defined as
\begin{equation} \label{eq:fidelity}
    F(\rho,\sigma) = \left(\text{Tr}\left(\sqrt{\sqrt\rho\sigma\sqrt\rho}\right)\right)^2,
\end{equation}
where $\sqrt\rho$ is well-defined due to $\rho$ being positive semi-definite, and $\sqrt\rho\sigma\sqrt\rho$ is also positive semi-definite by construction.
Observing additionally that the spectrum of a product of matrices is invariant under cyclic permutation, we recover
\begin{equation*}
    \text{Tr}\left(\sqrt{\sqrt\rho\sigma\sqrt\rho}\right) = \sum_{\lambda}\sqrt{\lambda} = \text{Tr}(\sqrt{\rho\sigma}),
\end{equation*}
where $\{\lambda\}$ is the set of eigenvalues of both $\sqrt\rho\sigma\sqrt\rho$ and $\rho\sigma$.
We can now define the \textit{quantum correlation}
\footnote{Strictly speaking, the operation over $\mathcal{R}$ is a $\sup$ operation. However, by restricting ourselves to finite-dimensional Hilbert spaces, the optimizations are being taken over compact sets and is equivalent to a $\max$ operation.}
as in \cite{Konig_2009}:
\begin{equation}\label{eq:q_corr}
    q_{corr} = d_R\max_\mathcal{R}F((I_R \otimes \mathcal{R})(\rho_{RE}), \ket{\Phi_{RE}}\bra{\Phi_{RE}}).
\end{equation}

Following \cite{Konig_2009}, we define $H_{min}$ in terms of $q_{corr}$:
\begin{equation*}
    H_{min}(R|E)_\rho = -\log q_{corr}(R|E)_\rho.
\end{equation*}
Because $E$ represents our experimental system post-error, the maximization in Eq. (\ref{eq:q_corr}) returns the maximum achievable fidelity between our experimental system and the untouched reference system.
Thus, $q_{corr}(R|E)$ quantifies how ``recoverable'' a state $E$ is by measuring its maximum fidelity with a perfectly recovered state $R$, and similarly $H_{min}(R|E)$ represents how much information is missing from the state.

The conditional min-entropy establishes the maximum fidelity between the post-recovery experimental state and the reference state, measuring how much information was leaked to the environment and therefore unrecoverable.
Crucially, this single-shot measure can be applied to any QEC code (and therefore ramp QSS scheme).
By maximizing over all possible recovery channels, we are assuming that the adversary gains access to some optimal recovery channel even with intermediate shares.
Therefore, $q_{corr}$ represents how similar the adversary's recovered state is to the perfectly recovered state under an optimal recovery, and $H_{min}$ represents the amount of information that is missing in their recovered secret compared to the reference secret. Notably, by taking the difference between the value of $H_{min}$ and its minimum possible value ($-\log K$) for a given code, the resulting quantity is proportional to the amount of information lost after erasure and recovery. In fact, Konig et al. \cite{Konig_2009} labels the conditional min entropy as the maximum achievable singlet fraction, inversely proportional to the ``similarity'' of the recovered state to the maximally entangled state.

As such, we use $H_{min}$ to quantify the amount of information leaked to the adversary as a value ranging between $[-\log K,\log K]$ that is proportional to this information leakage. As opposed to the Von Neumann conditional entropy (or mutual information), this quantity is single-shot and minimized, thereby representing the adversary's best case attempt at decoding the secret. The conditional min-entropy has similarly been used to quantify the security of quantum cryptographic protocols, including quantum key distribution~\cite{renner2008security}.

\subsection{Calculating $H_{min}$ for Stabilizer Codes}

Stabilizer codes represent a widely used class of QEC codes, thereby also functioning as common QSS schemes. We denote a stabilizer code with distance $d$ encoding $k$ logical qubits into $n$ physical qubits as an $\stabParam{n,k,d}$ code, which also makes it an $\codeParam{n,2^k,d}$ code.
Notably, stabilizers already have a measure quantifying the amount of information retrievable from intermediate sets that are neither fully authorized or fully vanishing given by the cleaning lemma for stabilizer codes \cite{Bravyi_2009}.
The cleaning lemma quantifies the number of logical Pauli operators that act nontrivially only on the erased partition of qubits ($M^c$) when multiplied by a stabilizer element $s \in S$, where $S$ is the code's stabilizer group, denoting such an operator as being ``cleaned off'' of $M$.
It follows then that the remaining Pauli operators must be ``cleaned off'' of the partition's complement set $M^c$ instead, and we denote the operators that are cleaned off $M^c$ ($M$) as being ``supported'' on $M$ ($M^c$ respectively).
Let $\bar{p} \in G_M$ if it is supported on $M$ and $\bar{p} \in G_{M^c}$ if it is instead supported on $M^c$.
Then, the logical information recoverable from each partition is
    \begin{equation} \label{eq:cleaning}
        \rho_M = \frac{1}{|G_{M^c}|}\sum_{E\in G_{M^c}}\!\!E\rho E^\dagger, \,\,\, \rho_{M^c} = \frac{1}{|G_M|}\sum_{F\in {G_M}}\!\!F\rho F^\dagger,
    \end{equation}
respectively, where $\rho$ is the original logical state \cite{Gheorghiu_2012}.
This makes it clear that the information recoverable using the adversary's shares $\rho_M$ vanishes with a greater number of operators that can be cleaned off of them. For any stabilizer code with encoding isometry $V: \mathbb{C}^{2^k} \rightarrow \mathbb{C}^{2^n}$, $G_M$ can be defined in terms of the normalizer and stabilizer groups as
\[G_M = V(N_M(S)/ S_M) V^{\dagger},\]
where $N_M(S)$ is the set of normalizer elements that are supported on $M$, and $S_M$  are the set of stabilizer elements that are supported on $M$. This can be observed from the fact that the logical operators of the code are defined as $N(S) / S$, and therefore $G_M$ is the subgroup of these operators that are also supported on $M$. We include the decoding isometry $V^{\dagger}$ to denote that the operators in $G_M$ act on the logical subspace.

While it is possible to calculate the conditional min-entropy by solving the optimization problem in Eq. \ref{eq:q_corr}, in general this is not tractable for larger codes.
Finding the optimal recovery channel scales exponentially with the number of qubits, causing the runtime to quickly become infeasible.
While there is not currently a general solution for this exponential runtime for all types of codes, stabilizer codes can make use of certain maximum likelihood (ML) decoders that represent their optimal recovery channels.
ML decoders are generally NP-complete to compute \cite{Hsieh_2011, DecodingHardnessIEEE}, but in the case of the erasure channel, they can be computed in polynomial time in terms of the erasure weight and number of stabilizers \cite{Harris_2018}.
In such cases, Eq. \ref{eq:cleaning} can be leveraged to obtain the following reduced form of the maximum fidelity:

\begin{theorem}[$H_{min}$ for Stabilizer Codes\footnote{This result previously appeared in~\cite{9174340}, although in that paper it was studied in terms of the maximum fidelity and not in terms of the conditional min-entropy.}]\label{th:stab_fidelity} \hfill\par
Let $\mathcal{C}$ be an $\stabParam{n,k,d}$ stabilizer code with orthonormal basis states $\ket{i}$. If $M^c \subseteq E$ is the set of erased qubits, the remaining state after optimal recovery channel $\mathcal{R}$ is $\rho_M = (I_R \otimes \mathcal{R}) (\Tr_{M^c}(\ketbra{\Phi_{RE}}{\Phi_{RE}}))$, where $\ket{\Phi_{RE}} = \frac{1}{\sqrt{K}}\sum_{i}\ket{i}_R\ket{i}_E$.
Then, the maximum fidelity is
\begin{equation*}
    F_{max}(\ket{\Phi_{RE}}, \rho_M) = \frac{1}{|G_{M^c}|},
\end{equation*}
and the corresponding $H_{min}$ value is
\begin{equation} \label{eq:HminStab}
    H_{min}(R | E) = -\log \frac{2^k}{|G_{M^c}|} = \log|G_{M^c}| - k.
\end{equation}
\end{theorem}

\begin{proof}
First, we consider the expression for fidelity in the case of one of the states being pure. This can be seen by letting $\rho$ be pure in Eq. \ref{eq:fidelity} such that $\rho = \ketbra{\psi_\rho}{\psi_\rho}$ and noting therefore that $\sqrt{\rho} = \rho$. Applying this to Eq. \ref{eq:fidelity},
\begin{align*}
    F(\rho,\sigma) &= \braket{\psi_\rho|\sigma|\psi_\rho}.
\end{align*}
Because $\ket{\Phi_{RE}}$ is a pure state in the calculation of $q_{corr}$ (Eq.~\ref{eq:q_corr}), we can apply this formulation of the fidelity to reveal
\[F(\ket{\Phi_{RE}}, \rho_M) = \braket{\Phi_{RE} | \rho_M | \Phi_{RE}},\]
where $\rho_M$ is the logical combined state $\rho = \ketbra{\Phi_{RE}}{\Phi_{RE}}$ after encoding, erasure, and recovery.
Now, by applying the result of Eq. \ref{eq:cleaning}, we can write $\rho_M$ in terms of $\rho$:
\begin{align*}
    \rho_M &= \frac{1}{|G_{M^c}|}\sum_{F\in G_{M^c}} (I_R \otimes F)\rho (I_R \otimes F)^\dagger, \\
    &= \frac{1}{|G_{M^c}|}\sum_{F\in G_{M^c}} (I_R \otimes F)\ketbra{\Phi_{RE}}{\Phi_{RE}} (I_R \otimes F)^\dagger.
\end{align*}
We can then write the fidelity as
\begin{align*}
    F(\ket{\Phi_{RE}}, \rho_M) &= \frac{1}{|G_{M^c}|}\sum_{F\in G_{M^c}}|\braket{\Phi_{RE} | I_R \otimes F | \Phi_{RE}}|^2.
\end{align*}
Because $\ket{\Phi_{RE}} = \frac{1}{\sqrt{2^k}}\sum_{i}\ket{i}_R\ket{i}_E$ represents an orthonormal basis for $\mathbb{C}^{2^k}$, the inner products in the summation reduce to $\frac{1}{2^k}\Tr(F)$, causing all summation terms to vanish except for the case where $F = I_E$.
Therefore, the expression for fidelity of a stabilizer code reduces to
\[F(\ket{\Phi_{RE}}, \rho_M) = \frac{1}{|G_{M^c}|}.\]
We can now find both $q_{corr}$ and subsequently the expression for $H_{min}(R|E)$ as
\begin{align*}
        H_{min}(R|E) &= -\log \frac{2^k}{|G_{M^c}|} = \log|G_{M^c}| - k. \qedhere
\end{align*}
\end{proof}

To show this explicitly, we will analyze a 4-qubit code and calculate its conditional min-entropy to showcase the operational meaning of the quantity for QSS schemes.
\begin{example}
    We present the encoded logical Paulis and stabilizer generators for the $\stabParam{4,1,2}$ Leung-Nielsen-Chuang-Yamamoto code \cite{Leung_1997} below:
\begin{align*}
    \bar{X} = XXII, \,\, \bar{Z} = ZIZI, \,\, S = \left\langle XXXX, ZZII, IIZZ \right\rangle.
\end{align*}
This code has certain known properties about its information leakage. 
Due to the distance, we would expect no loss of information up to a single erasure, but we also observe that the code completely loses all information with $n-d+1 = 3$ or more qubits erased. This is because complete loss of information on one partition implies the information is stored in the complementary partition. In the case where two qubits are lost, we can first observe the logical operators in $G_{M^c}$ assuming qubits 1 and 2 are erased:
\[G_{M^c} = \{X, I\}.\]
These are two out of the four possible logical operators, and inserting into Eq. \ref{eq:HminStab} gives 
\[H_{min}(R|E) = \log2 - 1 = 0,\]
denoting that half of the information has been lost. Stabilizer codes are subject to different permutations of the same weight erasure, but in this case the $H_{min}$ values (and therefore amount of information lost) does not change. For example, if we instead erase qubits 2 and 4, then $G_{M^c} = \{Z,I\}$ and we retrieve the same $H_{min}$ value of 0.

Table \ref{tab:4_codes} reflects our assumptions, with complete information leakage with more than three shares, and no information leakage with one or fewer shares. In all cases, the value of $F$ that corresponds to complete information loss is $0.25$. As explained in Section \ref{subsec:cond-min-entropy}, this matches the expected maximum fidelity in the case of complete loss of information is $1/{d_R^2}$, where $d_R = 2$. Notably, the value of $H_{min}$ with two shares indicates that the adversary's attempt at recovery obtained half the information of the true secret, matching the result intuited from the cleaning lemma. 

We can also directly observe our result given in Theorem \ref{th:stab_fidelity} as, for authorized share sets $M$, $G_{M^c}$ contains only the identity, and we therefore obtain a fidelity of $1$ and an $H_{min}$ value of $-1$. Conversely, for a vanishing share set, $G_{M^c}$ contains all four single qubit Paulis, therefore leading to a maximum fidelity of $1/4$ and $H_{min}$ value of 1. As discussed previously, with two missing shares, $G_{M^c}$ contains half of the logical Paulis, and therefore the maximum fidelity equals $1/2$, leading to an $H_{min}$ value of $0$. By utilizing ML decoders along with the cleaning lemma for stabilizers, we skip the expensive optimization normally required to calculate $q_{corr}$ and instead obtain a direct expression for maximum fidelity in terms of $|G_{M^c}|$. Additionally, we show that the conditional min-entropy is inversely proportional to the number of operators (and therefore bits of information) that are cleaned off of the shares held by the adversary.

\end{example}

\section{Quantum Anonymous Secret Sharing Protocol}

\subsection{QASS Protocol}
\begin{protocol}
    \caption{\textit{QASS}}
    Prerequisite: An $\codeParam{n,2,d}$ PI code has encoded the secret $\ket\psi$ into $n$ shares $\ket{\phi_i}$ and has been distributed to the $n$ shareholders.
    \begin{algorithmic}[1]
        \State For each of the $k$ shareholders participating in recovery: \vspace{0.2em}
        \begin{subnumbers}
            \State Confirm only one player is currently sending their share using the anonymous collision detection protocol ($\lceil \log n\rceil + 1$ $n$-qubit GHZ states).\vspace{0.2em}
            \State Generate 3 copies of the $n$-qubit GHZ state among the $n$ players.\vspace{0.2em}
            \State Perform \textit{ANONQ($\ket{\phi_i}$)} using the shared states, sending to known receiver Bob.\vspace{0.2em}
        \end{subnumbers}
        
        \State Bob decodes with collected shares $\ket{\phi_1},\ket{\phi_2},...\ket{\phi_k}$ to retrieve secret $\ket{\psi}$.
    \end{algorithmic}
    \label{prot:QASS}
\end{protocol}

We assume that both the encoding of the secret into $n$ shares and the distribution has already been done. The encoding happens prior to distribution and therefore doesn’t interfere with any anonymity constraints, and the distribution of shares does not need to necessarily be anonymous for our protocol’s usage. Let $k$ be the number of shares required in order to decode the secret. We maintain that anonymity is preserved throughout the decoding process. Each of the $k$ share-holders use the \textit{ANONQ} (Protocol \ref{prot:ANONQ}) to transmit their shares to a designated decoder Bob. We provide a visualization of the \textit{ANONQ} protocol in Fig. \ref{fig:anonq}.  Note that, as per our \textit{AE} protocol, we maintain only sender-anonymity, not receiver-anonymity as the senders retain knowledge of Bob.

\begin{figure}[!t]
    \centering
    \includegraphics[width=.8\columnwidth]{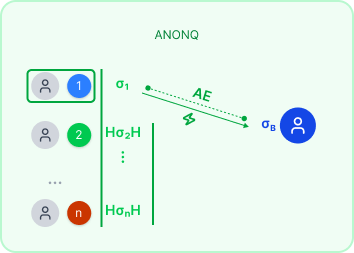}
    \setlength{\belowcaptionskip}{-1em}
    \caption{Overview of \textit{ANONQ($\ket{\phi}$)} (Protocol \ref{prot:ANONQ}) with Alice encircled in green and Bob in blue. Each participant holds one qubit of the $n$-qubit GHZ state, denoted as $\sigma_i$, and all participants besides Alice and Bob apply an $H$ gate to their qubit, leaving only Alice and Bob with an entangled pair. The green arrow represents the teleportation step after the \textit{AE} protocol.}
    \label{fig:anonq}
    \setlength{\belowcaptionskip}{0em}
\end{figure}
Next, Bob decodes the secret using an error syndrome consisting purely of the number of missing shares. This uses the unique property of PI codes in order for Bob to remain ignorant of the share-holders’ identities while retrieving the secret. The secret decoding is also traceless, as once Bob has the shares there is no more communication required to obtain the secret. Additionally, the adversary's knowledge of Bob's randomness does not identify any of the senders due to the decoding being permutation-invariant and therefore ignorant of the senders' order or identities.

Because each individual step leaves the system in a state where the share-holders’ identities are unknown and none of the contents of each step reveals any information of the share-holders’ identities, this combined protocol maintains share-holder anonymity. This protocol also inherits the assumptions made by each step, notably we assume that all share-holders and Bob are honest, otherwise any one of the steps described above could be disrupted. This restriction, however, is only for correctness of the output of each protocol. Anonymity is preserved up to $n-2$ corrupted players and the protocol is traceless due to each sub-protocol being traceless. We summarize the protocol in Protocol \ref{prot:QASS} and Fig. \ref{fig:qass-protocol}.

\begin{figure*}[!t]
    \centering
    \includegraphics[width=.65\textwidth]{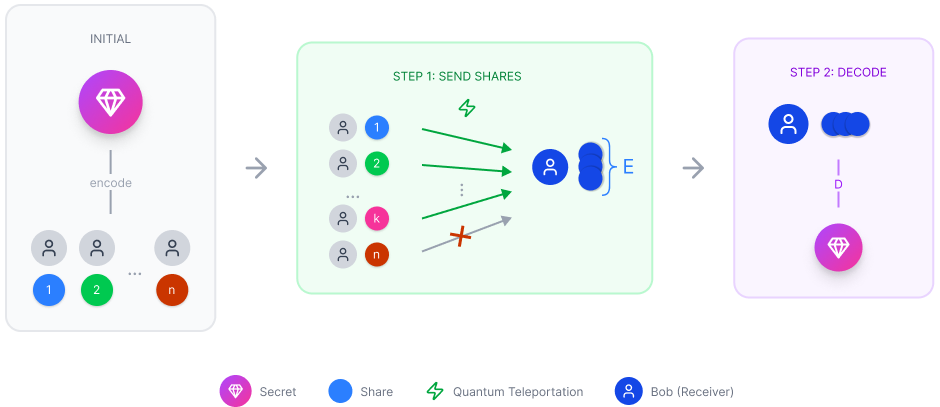}
    \caption{Overview of the \textit{QASS} protocol. We initially assume the secret has been encoded into $n$ shares. In Step 1, each green arrow denotes a share sent from each $k$ participating shareholders via the \textit{ANONQ} protocol. Once received by Bob, the identities of each share is entirely lost due to \textit{ANONQ}. Finally, Bob decodes the anonymous set of retrieved shares using the PI code's decoding circuit to obtain the secret.}
    \label{fig:qass-protocol}
\end{figure*}

\subsection{Hybrid \textit{QASS} Protocol}\label{sec:hyb-qass}
\begin{protocol}
    \caption{\textit{HQASS}}
    Prerequisite: A quantum secret $\ket\psi$ and two classical secrets $a$ and $b$ have been encoded using $\codeParam{n,2,d}$ PI codes into $n$ shares each ($\ket{\phi_i},\ket{a_i},\ket{b_i}$ respectively) and all shares have been distributed to the $n$ shareholders.
    \begin{algorithmic}[1]
        \State For each of the $k$ players participating in recovery:\vspace{.2em}
        \begin{subnumbers}
            \State Confirm only one player is currently sending their share using the anonymous collision detection protocol ($\lceil \log n\rceil + 1$ $n$-qubit GHZ states).\vspace{0.2em}
            \State Generate 9 copies of the $n$-qubit GHZ state among the $n$ players.\vspace{0.2em}
            \State Perform \textit{ANONQ($\ket{\phi_i}$)}, \textit{ANONQ($\ket{a_i}$)}, and \textit{ANONQ($\ket{b_i}$)} in sequence using the shared states, sending to known receiver Bob.\vspace{0.2em}
        \end{subnumbers}
        
        \State Bob decodes with  collected shares $\ket{a_1}...\ket{a_k}$ and $\ket{b_1}...\ket{b_k}$ to retrieve $a$ and $b$.\vspace{0.2em}
        \State Bob then decodes using $\ket{\phi_i}...\ket{\phi_k}$ to retrieve $X^aZ^b\ket{\psi}$ and applies the unitary transformations given by $a$ and $b$ to recover $\ket{\psi}$.
    \end{algorithmic}
    \label{prot:HQASS}
\end{protocol}
Protocol \ref{prot:QASS} can be easily modified to accommodate the hybrid scheme construction \cite{PhysRevA.64.042311,PhysRevA.71.012328} described in Section \ref{sec:sec-share}.
The dealer chooses two classical bits, $a$ and $b$, and applies a logical $X$ and $Z$ Pauli to the quantum secret depending on the value of the bits:
\begin{equation*}
    \rho_{RE\overline{E}}^{a,b} =(I_R\otimes \overline{X}^a\overline{Z}^b)\ket{\Phi_{RE}}\!\bra{\Phi_{RE}}(I_R\otimes \overline{Z}^b\overline{X}^a).
\end{equation*}
It is clear that, without knowledge of $a$ and $b$, this state reduces to the maximally mixed state, leaving the adversary with no distinctive information. Now we can ``lift'' the security of the quantum secret to that of classical secrets by encoding both $a$ and $b$ as classical secrets. 

Let $\ket{\overline 0}$ and $\ket{\overline 1}$ be the two encoded basis states corresponding to classical bit values $0$ and $1$ respectively. The dealer then encodes $a$ and $b$ as additional secrets using the states $\rho_{A\overline{A}}^a$ and $\rho_{B\overline{B}}^b$ equal to either $\ketbra{\overline 0}{\overline 0}$ or $\ketbra{\overline 1}{\overline 1}$ according to the values of $a$ and $b$. As before, we represent the erased subsystem using $\overline A$ and $\overline B$ for each of these classical secrets. After erasing $\overline E$, $\overline A$, and $\overline B$, the adversary is left with the following state:
\begin{equation*}
    \frac{1}{4}\left(\sum_{a,b\in\mathbb{F}_2}\Tr_{R\overline{E}}\left(\rho^{a,b}_{RE\overline{E}}\right)\otimes\Tr_{\overline{A}}(\rho_{A\overline{A}}^a)\otimes\Tr_{\overline{B}}(\rho_{B\overline{B}}^b)\right).
\end{equation*}
We adapt our \textit{QASS} protocol for the hybrid case in Protocol \ref{prot:HQASS}.

\section{Information Leakage Results for QSS Schemes}
\subsection{Experimental Setup}

We use Python's \texttt{cvxopt} package to solve for $q_{corr}$ using a semi-definite program. In particular, we gather results for 4, 7, 9, 11, and 13 qubit codes listed in Tables \ref{tab:4_codes} and \ref{tab:7-13_codes}.
Some of these codes have been labeled as ``shifted'' codes, these are smaller gnu codes that have been shifted using parameter $\Delta$ to a larger physical subspace as described by Ouyang \cite{ouyang2021permutation}. For each code, we gather a series of data-points, one per participating shareholder, to determine how much information leaked by the share set. We label the total number of participating shareholders as $n_p = |E|$. For the task of QSS, less information remaining means lower amounts of leakage, and therefore a more secure state. To gather each individual point of data, we input a maximally mixed single qubit $\ket{\psi_{R}} = \frac{1}{\sqrt{2}}(\ket0+\ket1)$ as the reference state and similarly input a maximally mixed single qubit into the encoding circuit to return $\ket{\psi_{E\bar{E}}} = \frac{1}{\sqrt{2}}(\ket{\bar0}+\ket{\bar1})$ where $\ket{\bar0}$ and $\ket{\bar1}$ represent the physical $0$ and $1$ basis states after encoding $\ket{0},\ket{1} \rightarrow \ket{\bar0},\ket{\bar1}$ respectively. We then entangle the two systems to obtain the overall input pure state
\[\ket{\Phi_{RE\bar{E}}} = \frac{1}{\sqrt2}\left( \ket0\!\ket{\bar0} + \ket1\!\ket{\bar1}\right).\]

We note that this is exactly the maximally mixed state required by the second Knill-Laflamme condition (Eq. \ref{eq:Knill-Laflamme-2}) except using the logical basis state for the reference register. This is equivalent to the explicit condition due to the fact that encoding the reference system and then performing no action on it means that the mapping between the logical and encoded state for the reference system is one-to-one. Therefore using the logical state as the reference state is equivalent to using the encoded state. Using an initial density matrix of $\rho_{RE\bar{E}} = \ketbra{\Phi_{RE\bar{E}}}{\Phi_{RE\bar{E}}}$, where $R$ is the reference system, $E$ is the experimental system's qubits that remain after erasure, and $\bar{E}$ is the set of qubits that get erased. We then simulate $|\bar{E}|$ erasures by tracing out $\bar{E}$, giving $\rho_{RE} = \text{Tr}_{\bar{E}}(\rho_{RE\bar{E}})$.

According to the definition of conditional min-entropy and the second Knill-Laflamme condition presented in Section \ref{subsec:cond-min-entropy} and  Definition \ref{def:Knill-Laflamme} respectively, we should be calculating $H_{min}(\bar{R}|E\bar{E})$, where $\bar{R}$ represents the \textit{encoded} reference state. Additionally, $E\bar{E}$ does not appear to lose dimension, and therefore erasures should be modeled as depolarizing qubits from the experimental state rather than tracing them out. However, because the encoded basis states are isomorphic to the logical basis states, we can substitute these for the logical basis states instead: 
\[\ket{\Phi_{RE}} = \frac{1}{\sqrt{d}}\sum_{x\in\{0,1\}^d}\ket{x}_R\ket{x}_E,\]
which in our case of each subsystem consisting of 1 logical qubit simplifies to:
\[\ket{\Phi_{RE}} = \frac{1}{\sqrt2}(\ket{00} + \ket{11}).\]

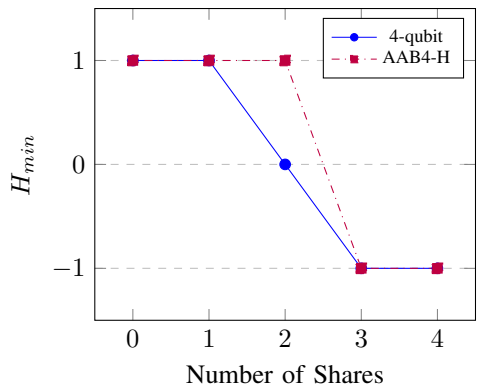
\begin{figure}[t]
\centering
\begin{tikzpicture}[trim axis left]
\begin{axis}[
    width=0.75\columnwidth,     xlabel={Number of Shares},
    ylabel={$H_{min}$},
    xmin=-0.5, xmax=4.5,
    ymin=-1.5, ymax=1.5,
    ymajorgrids=true,
    grid style=dashed,
    legend pos=north east,
    legend style={at={(0.97,0.97)}, anchor=north east, font=\scriptsize, row sep=-2pt, legend image post style={mark size=1.4pt, line width=0.3pt}}
]

\addplot[color=blue,mark=*] coordinates {
    (4,-1.0) (3,-1.0) (2,0.0) (1,1.0) (0,1.0)
};
\addlegendentry{4-qubit}
\addplot[dashdotted,color=purple,mark=square*] coordinates {
    (4,-1.0) (3,-1.0) (2,1.0) (1,1.0) (0,1.0)
};
\addlegendentry{AAB4-H}

\end{axis}

\end{tikzpicture}
\setlength{\belowcaptionskip}{-0.8em}
\caption{$H_{min}$ for QASS schemes using the $\codeParam{4,2,2}$ Aydin et al. PI code~\cite{Aydin_2024} and the $\codeParam{4,2,2}$ Hagiwara and Nakayama PI code~\cite{hagiwara2020fourqubitscodequantumdeletion}, the ramp QSS scheme using the $\stabParam{4,1,2}$ Leung et al. stabilizer code~\cite{Leung_1997}, and the HQASS scheme using Aydin et al. PI code~\cite{Aydin_2024} for both quantum and classical secrets. All three 4-qubit QSS schemes have identical $H_{min}$ values, whereas the HQASS scheme has a threshold of $k=3$ with no intermediate share sets.} 
\label{fig:4-qub}
\setlength{\belowcaptionskip}{0em}
\end{figure}
Note that this is only the objective state, and that $\rho_{RE}$ is still the bipartite state between the entangled logical reference system and encoded experimental system post-erasure(s). Therefore, the quantum channel $\mathcal{E}$ that the SDP will be optimizing over can be seen as the optimal recovery channel from some set of shares post-erasure to a single qubit state representing the adversary's attempt at the decoded secret. If the adversary has a fully authorized set of qubits, then this will be the exact secret, representing full leakage of information.

We can now use $\rho_{RE}$ as described in the semi-definite primal program for calculating $q_{corr}$ and thereby retrieve both $H_{min}(R|E)$ and the maximum fidelity $F$ (by dividing $q_{corr}$ by $d=\text{dim}(R)= 2$). We consider the maximum fidelity as a more stable measure of similarity as it remains a constant value between $0$ and $1$ denoting the similarity between the maximally recovered state of $X$ after solving the SDP and the maximally entangled state $\ket{\Phi_{RE}}$ necessary for total recovery. In particular, we expect the values of both $H_{min}$ and maximum fidelity to be the same between all three 4-qubit codes after a single erasure and after three erasures. This is because after a single erasure, all of the 4-qubit codes remain fully recoverable and afterthreeerasures all of the 4-qubit code resolve to a single maximally mixed qubit.

We also calculate the $H_{min}$ values for a HQASS scheme using the 4-qubit Aydin et al. code and the setup described in Section \ref{sec:hyb-qass}. We chose this code because when using the computational basis, it acts as a perfect CSS scheme, however we add that the 4-qubit Hagiwara and Nakayama code is also a perfect CSS scheme when encoded in the logical $Y$ basis ($\ket{\overline{\pm i}} = \frac{1}{\sqrt 2}(\ket{\overline 0} \pm i\ket{\overline 1})$). In this setting, we assume each shareholder either participates fully or not at all, meaning $n_p = |E| = |A| = |B|$, where $A$ and $B$ represent the present shares for secrets $a$ and $b$ respectively.

\subsection{Results}

Tables \ref{tab:4_codes} and \ref{tab:7-13_codes} shows how information remaining on the secret decreases nonlinearly based on the number of remaining qubits. Notably, all three of the 4-qubit codes of distance 2 measured the same $H_{min}$ and $F$ values after each number of qubits erased. As such, we combine the results for all 4-qubit codes into a single figure (Fig. \ref{fig:4-qub}). As shown in Fig. \ref{fig:combined-qubits}, the amount of information gained per share is not necessarily constant, with Fig. \ref{fig:13-qub} showing certain shares gathering no new information (denoted by the several plateaus of $H_{min}$).
\begin{table}[ht]
    \centering
    \caption{Summary of $H_{\text{min}}$ and maximum fidelity for 4-qubit codes.}
    \label{tab:4_codes}
    \vspace{1em}
    \resizebox{\columnwidth}{!}{    \begin{tabular}{@{}c | c c | c c | c c  | c c@{}}
        \toprule
        \multirow{2}{*}{$n_p$} & \multicolumn{2}{c|}{\textbf{AAB4~\cite{Aydin_2024}}} & \multicolumn{2}{c|}{\thead{\textbf{HN4~\cite{hagiwara2020fourqubitscodequantumdeletion}}}} & \multicolumn{2}{c|}{\textbf{LNCY4 \cite{Leung_1997}}} & \multicolumn{2}{c}{\thead{\textbf{AAB4-H~\cite{Aydin_2024}}}} \\
        \cmidrule{2-9}
        & $H_{\text{min}}$ & $F$ & $H_{\text{min}}$ & $F$ & $H_{\text{min}}$ & $F$ & $H_{\text{min}}$ & $F$ \\
        \midrule
        4 & -1.0 & 1.0 & -1.0  & 1.0 & -1.0 & 1.0 & -1.0 & 1.0 \\
        3 & -1.0  & 1.0 & -1.0 & 1.0 & -1.0 & 1.0 & -1.0 & 1.0 \\
        2 & 0.0  & 0.5 & 0.0 & 0.5 & 0.0 & 0.5 & 1.0 & 0.25 \\
        1 & 1.0  & 0.25 & 1.0 & 0.25 & 1.0 & 0.25 & 1.0 & 0.25 \\
        0 & 1.0  & 0.25 & 1.0 & 0.25 & 1.0 & 0.25 & 1.0 & 0.25 \\
        \bottomrule
    \end{tabular}
    }
    
\end{table}

 \begin{table*}[ht]      
   \centering         
   \caption{Summary of $H_{\text{min}}$ and maximum fidelity for 7, 9, 11, and 13-qubit PI codes.}             
   \label{tab:7-13_codes}          
   \setlength{\tabcolsep}{3pt}     
   \begin{tabular}{@{}c | cc | cc | cc | cc | cc | cc | cc | cc @{}}   
   \toprule                                                                                    
   \multirow{2}{*}{\shortstack{Number of\\Shares}} &                                  
                \multicolumn{2}{c|}{\thead{\textbf{PR7~\cite{pollatsek2004permutationally}}}} &  
  \multicolumn{2}{c|}{\thead{\textbf{AAB7~\cite{Aydin_2024}}}} &          
  \multicolumn{2}{c|}{\textbf{R9~\cite{Ruskai_2000}}} & 
  \multicolumn{2}{c|}{{\thead{\textbf{KT11~\cite{Kubischta_2023}}}}} &  
  \multicolumn{2}{c|}{\thead{\textbf{O11 \cite{ouyang2021permutation}}}} &                
  \multicolumn{2}{c|}{\thead{\textbf{AAB11 \cite{Aydin_2024}}}} &
  \multicolumn{2}{c|}{\thead{\textbf{KT13 \cite{Kubischta_2023}}}} &
  \multicolumn{2}{c}{\thead{\textbf{O13 \cite{ouyang2021permutation}}}} \\
   \cmidrule(lr){2-3} \cmidrule(lr){4-5} \cmidrule(lr){6-7} \cmidrule(lr){8-9}      
  \cmidrule(lr){10-11} \cmidrule(lr){12-13} \cmidrule(lr){14-15} \cmidrule(l){16-17}                                          
   & $H_{\text{min}}$ & $F$ & $H_{\text{min}}$ & $F$ & $H_{\text{min}}$ & $F$ &      
  $H_{\text{min}}$ & $F$ & $H_{\text{min}}$ & $F$ & $H_{\text{min}}$ & $F$ & $H_{\text{min}}$ & $F$ & $H_{\text{min}}$ & $F$ \\        
   \midrule                                                                                    
   13 & -- & -- & -- & -- & -- & -- & -- & -- & -- & -- & -- & -- & -1.0 & 1.0 & -1.0 & 1.0 \\              
   12 & -- & -- & -- & -- & -- & -- & -- & -- & -- & -- & -- & -- & -1.0 & 1.0 & -1.0 & 1.0 \\              
   11 & -- & -- & -- & -- & -- & -- & -1.0 & 1.0 & -1.0 & 1.0 & -1.0 & 1.0 & -1.0 & 1.0 & -1.0 & 1.0 \\            
   10 & -- & -- & -- & -- & -- & -- & -1.0 & 1.0 & -1.0 & 1.0 & -1.0 & 1.0 & -0.96 & 0.97 & -0.92 & 0.95 \\         
   9 & -- & -- & -- & -- & -1.0 & 1.0 & -1.0 & 1.0 & -1.0 & 1.0 & -1.0 & 1.0 & -0.85 & 0.9 & -0.78 & 0.86 \\     
   8 & -- & -- & -- & -- & -1.0 & 1.0 & -0.93 & 0.96 & -0.87 & 0.91 & -0.82 & 0.88 & -0.52 & 0.72 & -0.60 & 0.76 \\            
   7 & -1.0 & 1.0 & -1.0 & 1.0 & -1.0 & 1.0 & -0.68 & 0.8 & -0.67 & 0.79 & -0.64 & 0.78 & -0.26 & 0.6 & -0.40 & 0.66 \\                                                                                                   
   6 & -1.0 & 1.0 & -1.0 & 1.0 & -0.77 & 0.85 & -0.23 & 0.59 & -0.4 & 0.66 & -0.38 & 0.65 & 0.04 & 0.48 & -0.21 & 0.58 \\                                                                                                   
   5 & -1.0 & 1.0 & -1.0 & 1.0 & -0.33 & 0.63 & 0.09 & 0.47 & -0.14 & 0.55 & -0.17 & 0.56 & 0.04 & 0.48 & 0.00 & 0.50 \\                                                                                                   
   4 & -0.32 & 0.62 & -0.32 & 0.62 & 0.09 & 0.47 & 0.39 & 0.38 & 0.19 & 0.44 & 0.15 & 0.45 & 0.46 & 0.36 & 0.31 & 0.40 \\                                                                                                   
   3 & 0.32 & 0.4 & 0.32 & 0.4 & 0.4 & 0.38 & 0.44 & 0.37 & 0.53 & 0.35 & 0.45 & 0.37 & 0.46 & 0.36 & 0.61 & 0.33 \\             
   2 & 1.0 & 0.25 & 1.0 & 0.25 & 1.0 & 0.25 & 1.0 & 0.25 & 1.0 & 0.25 & 1.0 & 0.25 & 1.0 & 0.25 & 1.0 & 0.25 \\            
   1 & 1.0 & 0.25 & 1.0 & 0.25 & 1.0 & 0.25 & 1.0 & 0.25 & 1.0 & 0.25 & 1.0 & 0.25 & 1.0 & 0.25 & 1.0 & 0.25 \\            
   0 & 1.0 & 0.25 & 1.0 & 0.25 & 1.0 & 0.25 & 1.0 & 0.25 & 1.0 & 0.25 & 1.0 & 0.25 & 1.0 & 0.25 & 1.0 & 0.25 \\ \bottomrule
   \end{tabular}
\end{table*}

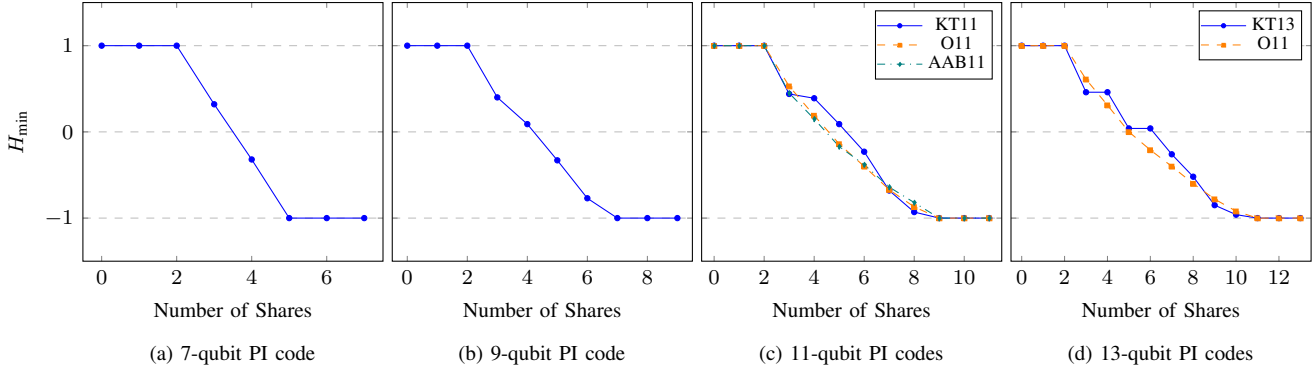
\begin{figure*}[!t]

 \centering

\begin{subfigure}{0.22\textwidth}
 
 \centering
 \begin{tikzpicture}[trim axis right, trim axis left, baseline]
 \begin{axis}[
 width=\linewidth, scale only axis,
 ylabel={$H_{\text{min}}$}, xlabel={Number of Shares},
 xmin=-0.5, xmax=7.5, ymin=-1.5, ymax=1.5, ymajorgrids=true, grid style=dashed,
 label style={font=\footnotesize},
 tick label style={font=\footnotesize},
 every axis plot/.append style={line width=0.4pt, mark size=1pt},
 major tick length=2pt,
     ]
 \addplot[color=blue,mark=*] coordinates {(7,-1.0) (6,-1.0) (5,-1.0) (4,-0.32) (3,0.32) (2,1.0) (1,1.0) (0,1.0)};
  \end{axis}
        \end{tikzpicture}
 \caption{ 7-qubit PI code}
 \label{fig:7-qub}
\end{subfigure}
\begin{subfigure}{0.22\textwidth}
 
 \centering
 \begin{tikzpicture}[trim axis right, trim axis left, baseline]
 \begin{axis}[
 width=\linewidth, scale only axis,
 yticklabels={}, xlabel={Number of Shares}, 
 xmin=-0.5, xmax=9.5, ymin=-1.5, ymax=1.5, ymajorgrids=true, grid style=dashed,
 label style={font=\footnotesize},
 tick label style={font=\footnotesize},
 every axis plot/.append style={line width=0.4pt, mark size=1pt},
 major tick length=2pt,
 ]
 \addplot[color=blue,mark=*] coordinates {(9,-1.0) (8,-1.0) (7,-1.0) (6,-0.77) (5,-0.33) (4,0.09) (3,0.4) (2,1.0) (1,1.0) (0,1.0)};
 \end{axis}
        \end{tikzpicture}
 \caption{ 9-qubit PI code}
 \label{fig:9-qub}
\end{subfigure}
\begin{subfigure}{0.22\textwidth}
 
 \centering
 \begin{tikzpicture}[trim axis right, trim axis left, baseline]
 \begin{axis}[
 width=\linewidth, scale only axis,
 xlabel={Number of Shares},
 xmin=-0.5, xmax=11.5, xtick distance=2, ymin=-1.5, ymax=1.5, yticklabels={}, ymajorgrids=true, grid style=dashed,
 label style={font=\footnotesize},
 tick label style={font=\footnotesize},
 every axis plot/.append style={line width=0.4pt, mark size=1pt},
 major tick length=2pt,
 legend style={at={(0.97,0.97)}, anchor=north east, font=\scriptsize, inner sep=0.5pt, row sep=-3pt, legend image post style={mark size=0.8pt, line width=0.3pt}},
 ]
 \addplot[color=blue,mark=*] coordinates {(11,-1.0) (10,-1.0) (9,-1.0) (8,-0.93) (7,-0.68) (6,-0.23) (5,0.09) (4,0.39) (3,0.44) (2,1.0) (1,1.0) (0,1.0)};
 \addlegendentry{KT11}
 \addplot[color=orange,mark=square*, dashed] coordinates {(11,-1.0) (10,-1.0) (9,-1.0) (8,-0.87) (7,-0.67) (6,-0.4) (5,-0.14) (4,0.19) (3,0.53) (2,1.0) (1,1.0) (0,1.0)};
 \addlegendentry{O11}
 \addplot[color=teal,mark=diamond*, dashdotted] coordinates {(11,-1.0) (10,-1.0) (9,-1.0) (8,-0.82) (7,-0.64) (6,-0.38) (5,-0.17) (4,0.15) (3,0.45) (2,1.0) (1,1.0) (0,1.0)};
 \addlegendentry{AAB11}
 \end{axis}
        \end{tikzpicture}
 \caption{ 11-qubit PI codes}
 \label{fig:11-qub}
\end{subfigure}
\begin{subfigure}{0.22\textwidth}
 \centering
 \begin{tikzpicture}[trim axis right, trim axis left, baseline]
 \begin{axis}[
 width=\linewidth, scale only axis,
 yticklabels={}, xlabel={Number of Shares}, 
 xmin=-0.5, xmax=13.5, ymin=-1.5, ymax=1.5, ymajorgrids=true, grid style=dashed, xtick distance=2,
 label style={font=\footnotesize},
 tick label style={font=\footnotesize},
 every axis plot/.append style={line width=0.4pt, mark size=1pt},
 major tick length=2pt,
 legend style={at={(0.97,0.97)}, anchor=north east, font=\scriptsize, inner sep=0.5pt, row sep=-3pt, legend image post style={mark size=0.8pt, line width=0.3pt}}
 ]
  \addplot[color=blue, mark=*] coordinates {(13,-1.0) (12,-1.0) (11,-1.0) (10,-0.96) (9,-0.85) (8,-0.52) (7,-0.26) (6,0.04) (5,0.04) (4,0.46) (3,0.46) (2,1.0) (1,1.0) (0,1.0)};
 
  \addplot[color=orange, mark=square*, dashed] coordinates {(13,-1.0) (12,-1.0) (11,-1.0) (10,-0.92) (9,-0.78) (8,-0.60) (7,-0.40) (6,-0.21) (5,0.0) (4,0.31) (3,0.61) (2,1.0) (1,1.0) (0,1.0)};

 \legend{KT13, O11}
 \end{axis}
        \end{tikzpicture}
 \caption{ 13-qubit PI codes }
 \label{fig:13-qub}
\end{subfigure}

\vspace{.7em} 

    \setlength{\belowcaptionskip}{-1em}

 \caption{ Comparison of $H_{\text{min}}$ across various PI codes of distance 3. (a) The $\codeParam{7,2,3}$ codes of Pollatsek and Ruskai~\cite{pollatsek2004permutationally, Kubischta_2023} and Aydin et al.~\cite{Aydin_2024}, which have the same $H_{min}$ values. (b) The $\codeParam{9,2,3}$ code of Ruskai~\cite{Ruskai_2000} (the gnu $\left(3,3,1\right)$-PI code~\cite{PhysRevA.90.062317}). (c) Three $\codeParam{11,2,3}$ codes including the one given by Kubischta and Teixeira~\cite{Kubischta_2023}, the shifted gnu $\left(3,3,1\right)$-PI code of Ouyang~\cite{ouyang2021permutation}, and the $\mathcal{Q}_{4,1,2,-1}$ code of Aydin et al.~\cite{Aydin_2024}. (d) Two $\codeParam{13,2,3}$ codes including one given by Kubischta and Teixeira \cite{Kubischta_2023} and the shifted gnu (3, 3, 1)-PI code of Ouyang \cite{ouyang2021permutation}.}
 \label{fig:combined-qubits}

 \setlength{\belowcaptionskip}{0em}
\end{figure*}

\section{Conclusion}

By mitigating single points of failure around the control of confidential information, quantum secret sharing remains an important primitive in quantum cryptography.
In this paper, we extend the theory of QSS by providing an anonymous protocol that, when paired with a PI code, allows for total sender anonymity within the QSS scheme.
By obscuring the senders' identities from both adversaries and other participants, we not only retain anonymity in the case of an unbounded adversary, but also support cases in which the sender(s) wish to obfuscate their participation in general such as in anonymous voting \cite{hao2010anonymous, vaccaro2007quantum}.
The motivation for anonymity in both scenarios is simple: a participant should be able to participate in a protocol without fear of coercion or external ramifications, even from other participants or the recoverer.
We also put forward conditional min-entropy as a measure for quantifying the leakage of information in ramp QSS schemes, including those derived from PI codes.

There are many interesting future directions regarding improvements to different parts of our protocol.
The first would be replacing the quantum anonymous transmissions protocols we use from~\cite{Christandl_2005} with other more recent protocols. One limitation of the protocols we use is their dependence on an $n$-qubit GHZ state, with some newer works instead using W states \cite{Thalacker_2021,lipinska2018anonymous}, rendering the protocols more resistant to noise. Brassard et al. \cite{brassard2007anonymous} introduced a verification stage for the GHZ state, removing the assumption of a prior shared $n$-qubit GHZ state, while also ensuring that the state is only destroyed given a malicious participant with exponentially small probability. More recently, Unnikrishnan et al. \cite{PhysRevLett.122.240501} improved upon this by providing anonymity even assuming a malicious source and an approximate anonymous teleportation protocol that can be run without cumbersome size-$n$ quantum circuits. 
Another possible improvement is to identify PI codes that exhibit less dramatic leakage for small intermediate sets, leading to a more gradual ramping structure in the graphs of $H_{min}$.
In the hybrid \textit{QASS} schemes, we elected to use the same PI code to encode both the quantum and classical secrets; however, this is not necessary, and it motivates a formal investigation into PI codes designed specifically for the transmission of classical information.

Another direction would be formalizing the connection between the ramp QSS schemes studied here and approximate QSS schemes~\cite{crepeau2005approximateQEC,PhysRevA.108.012425,elimelech2026asymptoticallygoodbosonicfock}, both of which allow for information leakage.
In the approximate QSS, a small amount of information leakage is allowed in unauthorized sets, and a small error on the reconstructed secret is allowed in authorized sets.
Similar to our usage of conditional min-entropy to quantify information leakage, approximate schemes make use of the maximum fidelity to quantify leakage.

Finally, while the conditional min-entropy provides a sound measure of information leakage, it is also generally expensive to compute, with a time complexity proportional to the dimension of $\rho_{RE}$.
Since this matrix scales exponentially with the number of qubits, calculating the conditional min-entropy can be impractical even for relatively small QEC codes.
As such this motivates the developments of techniques for speeding up the calculations of $H_{min}$ or bounding it, both for general codes as well as for special classes such as the PI codes.
One idea might be to use techniques from the theory of quantum weight enumerators in QEC to reduce the size of the SDP instances~\cite{10609451,munne2025sdpboundsquantumcodes}.

\section*{Acknowledgment}
The authors would like to thank Eric Chitambar for fruitful discussions.

\bibliographystyle{IEEEtran}
\bibliography{bib-template-abbv}

\end{document}